\documentclass[12pt]{article}
\usepackage{fullpage}
\usepackage{amsmath,amssymb,amsfonts,amsthm}
\usepackage{enumerate}
\usepackage{epsfig}
\usepackage[all,poly]{xy}
\usepackage{colortbl}

\usepackage{geometry}
\usepackage{tikz}
\usepackage{mathrsfs}

\newif\ifnotesw\noteswtrue
   {\ifnotesw\marginpar[\hfill\(\top\)]{\(\top\)}\fi}%
   {\ifnotesw\marginpar[\hfill\(\bot\)]{\(\bot\)}\fi}

\newcommand{\mnote}[1]%
    {\ifnotesw\marginpar%
        [{\scriptsize\begin{minipage}[t]{\marginparwidth}
        \raggedleft#1%
                        \end{minipage}}]%
        {\scriptsize\begin{minipage}[t]{\marginparwidth}
        \raggedright#1%
                        \end{minipage}}%
    \fi}

\newcommand{\ignore}[1]{}

\newcommand{\etal}{{\it et al.~}}

\newtheorem{theorem}{Theorem}
\newtheorem{corollary}[theorem]{Corollary}
\newtheorem{lemma}[theorem]{Lemma}

\newcommand{\iverson}[1]{\lbrack\!\lbrack #1 \rbrack\!\rbrack}

\newcommand{\lip}{\langle}
\newcommand{\rip}{\rangle}
\newcommand{\NN}{\mathbb{N}}
\newcommand{\CC}{\mathbb{C}}
\newcommand{\ZZ}{\mathbb{Z}}

\newcommand{\GG}{\mathcal{G}}

\newcommand{\bra}[1]{\lip #1 |}
\newcommand{\ket}[1]{| #1 \rip}

\newcommand{\ketbra}[2]{| #1 \rip\lip #2 |}

\newcommand{\cart}{\mbox{ $\Box$ }}

\DeclareMathOperator{\diag}{diag}

\DeclareMathOperator{\sgn}{sgn}

\newcommand{\tdp}{\tilde{\delta}_{+}}
\newcommand{\tdm}{\tilde{\delta}_{-}}
\newcommand{\tdpm}{\tilde{\delta}_{\pm}}
\newcommand{\jj}{\mathbf{1}}
\newcommand{\SG}{\Sigma}
\newcommand{\Alt}{\mathsf{Alt}}

\newcommand{\wedgep}[2]{\bigwedge^{#2} #1}
\newcommand{\xwedgep}[2]{\mbox{$\bigwedge^{#2} #1$}}

\newcommand{\veep}[2]{{#1}^{\{ #2 \}}}
\newcommand{\symp}[2]{{#1}^{\odot {#2}}}

\newcommand{\symdif}{\triangle}

\title{
Perfect State Transfer on Signed Graphs 
}
\author{
John Brown\\SUNY Potsdam
\and
Chris Godsil\\University of Waterloo 
\and
Devlin Mallory\\UC Berkeley
\and
Abigail Raz\\Wellesley College
\and
Christino Tamon\footnote{Contact author: tino@clarkson.edu}\\Clarkson University
}
\date{\today}
\begin{document}
\maketitle
\bibliographystyle{plain}

\begin{abstract}
We study perfect state transfer of quantum walks on signed graphs.
Our aim is to show that negative edges are useful for perfect state transfer.
Specific results we prove include:
\begin{itemize}
\item The signed join of a negative $2$-clique with any positive $(n,3)$-regular graph 
	has perfect state transfer even if the unsigned join does not. 
	Curiously, the perfect state transfer time improves as $n$ increases.

\item A signed complete graph has perfect state transfer if its positive subgraph is a regular 
	graph with perfect state transfer and its negative subgraph is periodic.
	This shows that signing is useful for creating perfect state transfer since no complete graph
	(except for the $2$-clique) has perfect state transfer.

\item The double-cover of a signed graph has perfect state transfer if the positive subgraph has
	perfect state transfer and the negative subgraph is periodic. 
	Here, signing is useful for constructing unsigned graphs with perfect state transfer.
\end{itemize}
Furthermore, we study perfect state transfer on a family of signed graphs called 
the exterior powers which is derived from a many-fermion quantum walk on graphs.
\end{abstract}


\section{Introduction}

The study of quantum walks on finite graphs is important in quantum computing due to its promise as an 
algorithmic technique orthogonal to the Hidden Subgroup and Amplitude Amplification paradigms. Strong 
quantum algorithms based on quantum walks have been discovered in the ensuing years for diverse problems
such as element distinctness, matrix product verification, triangle finding in graph, formula evaluation,
and others. But recently, quantum walks have also proved crucial as a universal quantum computational model 
(see Childs \cite{childs09}).

A property of quantum walks called perfect state transfer was originally studied by Bose \cite{bose03} 
in the context of information transfer in quantum spin chains. Christandl \etal \cite{cdel04,cddekl05} 
continued this investigation and proved strong results for various other graphs, notably the hypercubes. 
More recently, perfect state transfer was used by Underwood and Feder \cite{uf10} in the simulation
of universal quantum computation via quantum walks. This provides an alternative method to the graph 
scattering techniques used in \cite{childs09}.

The main problem in the study of perfect state transfer in quantum walks on graphs is
to characterize graphs which exhibit this property. Much of the recent progress along 
these lines is described in Godsil \cite{godsil-dm11}. Another related question is
to ask for operations on graphs which create perfect state transfer. Examples of such operations
include deleting edges \cite{bcms09}, adding self-loops \cite{clms09}, and using weights on edges
\cite{f06}. The latter work by Feder \cite{f06} is based on a many-particle quantum walk on graphs
where the particles are bosons.

In this work, we study perfect state transfer of quantum walks on {\em signed} graphs.
A signed graph is a graph whose edges are given $\pm 1$ weights. The literature on
signed graphs is vast (see Zaslavsky \cite{z-survey}). 
Our main goal is to understand the impact of negative edges on perfect state transfer. 
To the best of our knowledge, Pemberton-Ross and Kay \cite{pk11} provided the first indication
that negative weights might be useful for perfect state transfer via dynamic couplings.
We show that negative edges are useful for creating perfect state transfer on certain classes of graphs
even without dynamic couplings. 
More specifically, we study the effect of signed edges on graph products, joins, quotients, 
and on signed graphs with certain spanning-subgraph decomposition properties. 

For graph joins, we show that the join between a negatively signed $K_{2}$ with
any $3$-regular unsigned graph has perfect state transfer; in contrast, the unsigned join
lacks this perfect state transfer property. Curiously, the perfect state transfer time in
the signed join decreases as the size of the $3$-regular graph increases.
Using the spanning-subgraph decomposition property of signed graphs, we show examples of
signed complete graphs with perfect state transfer. This is in contrast to the known fact
that unsigned complete graphs have no perfect state transfer (but are merely periodic).
We also show the opposite effect by constructing {\em unsigned} graphs with perfect state 
transfer from double-coverings of signed graphs.

Finally, we consider an interesting graph operator called the exterior power.
We observe that this operator creates, in a natural way, signed graphs from unsigned graphs.
This operator was studied by Osborne \cite{o06} (and also by Audenaert \etal \cite{agrr07}) 
although not in the context of signed graphs. The exterior power of a graph is related to
a many-particle quantum walk on graphs where the particle are fermions. More specifically, we show 
that the exterior $k$th power of a graph $G$ is the {\em quotient} graph corresponding to a $k$-fermion
quantum walk on $G$. This complements the result of Bachman \etal \cite{bfflott12} showing that Feder's 
weighted graphs (see \cite{f06}) are quotient graphs corresponding to a $k$-boson quantum walk 
on an underlying graph.


\section{Preliminaries}

We describe some notation which will be used throughout the paper.
For a logical statement $S$, we use $\iverson{S}$ to mean $1$ if $S$ is true, and $0$ otherwise.
Given a positive integer $n$, the notation $[n]$ denotes the set $\{1,\ldots,n\}$.
We use $A \uplus B$ to denote the disjoint union of sets $A$ and $B$.
The identity and all-one matrices are denoted $I$ and $J$, respectively; the latter may not necessarily be square. 

The graph $G=(V,E)$ we consider will be finite, undirected, and connected.
The adjacency matrix $A(G)$ of $G$ is defined as $A(G)_{u,v} = \iverson{(u,v) \in E}$.
A graph $G$ is called $k$-regular if each vertex of $G$ has exactly $k$ adjacent neighbors.
We say a graph $G$ is $(n,k)$-regular if it has $n$ vertices and is $k$-regular.

Let $G$ and $H$ be two given graphs.
The {\em Cartesian product} $G \cart H$ of graphs $G$ and $H$ has the vertex set $V(G) \times V(H)$ 
where the edges are defined as follows.
The vertex $(g_{1},h_{1})$ is adjacent to the vertex $(g_{2},h_{2})$ 
if $g_{1} = g_{2}$ and $(h_{1},h_{2}) \in E(H)$ or if $(g_{1},g_{2}) \in E(G)$ and $h_{1}=h_{2}$.
The adjacency matrix is given by $A(G \cart H) = A(G) \otimes I + I \otimes A(H)$.
The {\em join} $G + H$ of graphs $G$ and $H$ is a graph whose complement is $\overline{G} \uplus \overline{H}$.
The adjacency matrix of the join is given by
\begin{equation}
A(G + H) = \begin{bmatrix} A(G) & J \\ J & A(H) \end{bmatrix}.
\end{equation}
The dimensions of the matrices $I$ and $J$ used above are implicit (but clear from context).
Most notation we use above are adopted from \cite{godsil-royle01}.

\paragraph{Signed graphs}
A {\em signed} graph $\Sigma = (G,\sigma)$ is a pair consisting of a graph $G=(V,E)$ and 
a signing map $\sigma: E(G) \rightarrow \{-1,+1\}$ over the edges of $G$. 
For notational convenience, we may on occasion use $G^{\sigma}$ in place of $(G,\sigma)$.

We call $G$ the underlying (unsigned) graph of $\Sigma$; we also use $|\Sigma|$ to denote this underlying graph. 
Two signed graphs $\Sigma_{1}$ and $\Sigma_{2}$ are called {\em switching equivalent}, denoted
$\Sigma_{1} \sim \Sigma_{2}$, if 
\begin{equation}
A(\Sigma_{1}) = D^{-1}A(\Sigma_{2})D,
\end{equation}
for some diagonal matrix $D$ with $\pm 1$ entries.
A signed graph $\Sigma$ is {\em balanced} if $\Sigma \sim |\Sigma|$ and it is called
{\em anti-balanced} if $\Sigma \sim -|\Sigma|$ where $-|\Sigma|$ refers to the all-negative
signing of $|\Sigma|$.

Another way to view a signed graph $\Sigma$ is as two edge-disjoint spanning subgraphs of
the underlying graph $|\Sigma|$; that is, $\Sigma = G^{+} \cup G^{-}$, where the edges of
$G^{+}$ are signed with $+1$ and the edges of $G^{-}$ are signed with $-1$. 
The adjacency matrix of $\Sigma$ is then given as $A(\Sigma) = A(G^{+}) - A(G^{-})$. 

A classic reference on signed graphs is the paper by Zaslavsky \cite{z82}.

\paragraph{Quantum walks}
Given a graph $G$, a continuous-time quantum walk on $G$ is described by the time-dependent unitary matrix
\begin{equation}
U(t) = \exp(-it A(G)).
\end{equation} 
We say a graph $G$ has {\em perfect state transfer} from vertex $a$ to $b$ at time $t$ if 
\begin{equation}
|\bra{b}U(t)\ket{a}| = 1.
\end{equation} 
On the other hand, the graph $G$ is {\em periodic} at vertex $a$ at time $t$ if $|\bra{a}U(t)\ket{a}| = 1$. 
Finally, $G$ is called periodic if it is periodic at all of its vertices.
For more background on state transfer on graphs, we refer the reader to Godsil \cite{godsil-dm11}.


\section{Balanced Products}

We state some basic results for perfect state transfer on balanced (and anti-balanced) signed graphs.

\begin{lemma} \label{lemma:balanced}
If a graph $G$ has perfect state transfer, then so does the signed graph $\Sigma = (G, \sigma)$ 
provided $\sigma$ is a balanced or anti-balanced signing of $G$.
\end{lemma}
\begin{proof}
Suppose $G$ has perfect state transfer from vertex $a$ to $b$. If $\sigma$ is a balanced 
or anti-balanced signing of $G$, then there is a diagonal $\pm 1$ matrix $D$ for which 
$A(G^{\sigma}) = \pm D^{-1}A(G)D$. Thus, we have
\begin{equation}
\bra{b}e^{-it A(G^{\sigma})}\ket{a} = \bra{b}D^{-1}e^{\mp itA(G)}D\ket{a} 
	= \pm \bra{b}e^{\mp itA(G)}\ket{a}.
\end{equation}
This shows that $G^{\sigma}$ has perfect state transfer from $a$ to $b$.
\end{proof}

\begin{figure}[t]
\begin{center}
\begin{tikzpicture}[scale=1.0]
\foreach \x in {0,180}
{
    \node at (\x:1)[circle,fill=black][scale=0.6] {};
}
\foreach \x in {0,90,180,270}
{ 
	\draw[line width=0.2mm] (\x:1)--(\x+90:1);
}
\foreach \x in {90,270}
{
    \node at (\x:1)[circle,fill=white][scale=0.6] {};
    \draw[line width=0.2mm] (\x:1) circle (0.1cm);
}
\end{tikzpicture}	\quad	\quad
\begin{tikzpicture}[scale=1.0]
\foreach \x in {0,180}
{
    \node at (\x:1)[circle,fill=black][scale=0.6] {};
}
\foreach \x in {0,180}
{ 
	\draw[line width=0.2mm] (\x:1)--(\x+90:1);
}
\foreach \x in {90,270}
{ 
	\draw[dashed] (\x:1)--(\x+90:1);
}
\foreach \x in {90,270}
{
    \node at (\x:1)[circle,fill=white][scale=0.6] {};
    \draw[line width=0.2mm] (\x:1) circle (0.1cm);
}
\end{tikzpicture}	\quad	\quad
\begin{tikzpicture}[scale=1.0]
\foreach \x in {0,180}
{
    \node at (\x:1)[circle,fill=black][scale=0.6] {};
}
\foreach \x in {180,270}
{ 
	\draw[line width=0.2mm] (\x:1)--(\x+90:1);
}
\foreach \x in {0,90}
{ 
	\draw[dashed] (\x:1)--(\x+90:1);
}
\foreach \x in {90,270}
{
    \node at (\x:1)[circle,fill=white][scale=0.6] {};
    \draw[line width=0.2mm] (\x:1) circle (0.1cm);
}
\end{tikzpicture}	\quad	\quad
\begin{tikzpicture}[scale=1.0]
\foreach \x in {0,90,180,270}
{
    \node at (\x:1)[circle,fill=black][scale=0.6] {};
}
\foreach \x in {0,90,180}
{ 
	\draw[dashed] (\x:1)--(\x+90:1);
}
\foreach \x in {270}
{ 
	\draw[line width=0.2mm] (\x:1)--(\x+90:1);
}
\end{tikzpicture}
\caption{Signed $4$-cycles: (a) Unsigned; (b) Balanced; (c) Antibalanced; (d) Unbalanced.
Dashed edges are negatively signed. 
Perfect state transfer occurs between vertices marked white within each case.
}
\label{fig:balanced-q2}
\end{center}
\hrule
\end{figure}
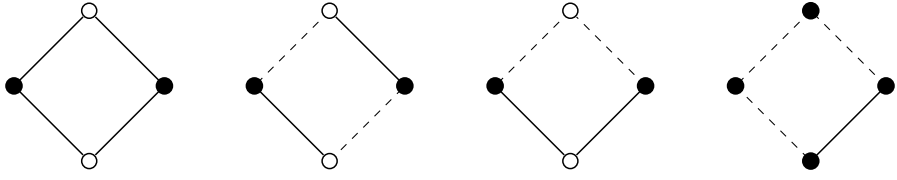

\begin{corollary}
For a positive integer $m$,
if for each $k \in [m]$, the graph $G_{k}$ has perfect state transfer from vertex $a_{k}$ to vertex $b_{k}$,
then the signed graph $\Box_{k=1}^{m} \Sigma_{k}$, where $\Sigma_{k} = (G_{k}, \sigma_{k})$, 
has perfect state transfer from vertex $(a_{1},\ldots,a_{m})$ to vertex $(b_{1},\ldots,b_{m})$,
provided each $\sigma_{k}$ is a balanced or anti-balanced signing of $G_{k}$.
\end{corollary}
\begin{proof}
We note that $A(\Box_{k=1}^{m} G_{k})$ is a sum of $k$ commuting terms:
\begin{equation}
A(\Box_{k=1}^{m} G_{k}) 
= 
\sum_{k=1}^{m} 
		(I \otimes \ldots \otimes I \otimes 
			\overbrace{A(G_{k})}^{\mbox{\scriptsize $k$th position}} 
			\otimes I \otimes \ldots \otimes I).
\end{equation}
Thus, we have
\begin{equation}
\bigotimes_{j=1}^{m} \bra{b_{j}} \left[ e^{-it A(\Box_{k} G_{k})}\right]  \bigotimes_{\ell=1}^{m} \ket{a_{\ell}}
=
\prod_{k=1}^{m} \bra{b_{k}} e^{-it A(G_{k})} \ket{a_{k}}.
\end{equation}
Now, we apply Lemma \ref{lemma:balanced}, using
$A(G_{k}^{\sigma_{k}}) = \pm D_{k}^{-1} A(G_{k}) D_{k}$, to obtain:
\begin{equation}
\bigotimes_{j=1}^{m} \bra{b_{j}} e^{-it A(\Box_{k} \Sigma_{k})} \bigotimes_{\ell=1}^{m} \ket{a_{\ell}}
	= \prod_{k=1}^{m} \bra{b_{k}} D_{k}^{-1} e^{\mp it A(G_{k})} D_{k} \ket{a_{k}} 
	= \pm \prod_{k=1}^{m} \bra{b_{k}} e^{\mp it A(G_{k})} \ket{a_{k}},
\end{equation}
since $D_{k}\ket{a} = \pm\ket{a}$ for each vertex $a$. This proves the claim.
\end{proof}


\section{Signed Joins}

In this section, we consider perfect state transfer properties of binary graph joins 
where the two graphs are given opposite signs. So, we denote $G^{+} \pm H^{-}$
to mean the join of graphs $G$ and $H$ where $G$ is positively signed, $H$ is negatively
signed, and their connecting edges are all positively (or negatively) signed, respectively.
We remark that $G^{+} \pm H^{-}$ are switching equivalent and hence share perfect state
transfer properties.

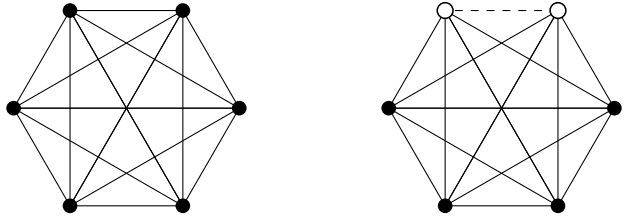
\begin{figure}[t]
\begin{center}
\begin{tikzpicture}[scale=1.5]
\foreach \x in {60,120,...,360} 
    \node at (\x:1)[circle,fill=black][scale=0.5] {};

\foreach \x in {60,120,...,360}
	\foreach \y in {60,120,180}
	{
		\draw (\x:1)--(\x+\y:1);
	}
\end{tikzpicture}	\quad	\quad	\quad	\quad
\begin{tikzpicture}[scale=1.5]
\foreach \x in {60,120,...,360} 
    \node at (\x:1)[circle,fill=black][scale=0.5] {};

\foreach \x in {60,120,...,360}
	\foreach \y in {120,180}
	{
		\draw (\x:1)--(\x+\y:1);
	}

\foreach \x in {120,180,...,360}
	\foreach \y in {60}
	{
		\draw (\x:1)--(\x+\y:1);
	}

\draw[dashed] (60:1)--(120:1);

\foreach \x in {60,120} 
{
	\node at (\x:1)[circle,fill=white][scale=0.5] {};
	\draw[line width=0.2mm] (\x:1) circle (0.07cm);
}
\end{tikzpicture}
\caption{Signed complete graphs. 
(a) $K_{6}$ has no perfect state transfer.
(b) The signed join $K_{2}^{-} + K_{4}^{+}$ has perfect state transfer between vertices marked white.
	The dashed edge is negatively signed.
}
\label{fig:signed-k6}
\end{center}
\hrule
\end{figure}

\begin{theorem} \label{thm:join}
Suppose $G_1$ is a $(n_1,k_1)$-regular graph and $G_2$ is a $(n_2,k_2)$-regular graph.
Let $\SG$ be the signed join graph $G_1^{-} + G_2^{+}$ 
Then, for two vertices $a,b \in V(G_1)$, we have
\begin{equation}
\bra{b}e^{-itA(\SG)}\ket{a} = \bra{b}e^{itA(G_1)}\ket{a} +
	\frac{e^{-it\tdm}}{n_1}
	\left[
	\left(\cos(t\Delta) - i \frac{\tdp}{\Delta} \sin(t\Delta)\right) - e^{-it\tdp}
	\right]
\end{equation}
where $\tdpm = -\frac{1}{2}(k_1 \pm k_2)$ and
$\Delta = \sqrt{\tdp^{2} + n_1 n_2}$.
\end{theorem}
\begin{proof}
For notational convenience, given $\ket{u}$ and $\ket{v}$ of dimensions $n_1$ and $n_2$, respectively,
let $\ket{u,v}$ denote the $(n_1 + n_2)$-dimensional ``concatenated'' column vector whose projection
onto the first $n_1$ dimensions is $\ket{u}$ and whose projection onto the last $n_2$ dimensions is $\ket{v}$.

The adjacency matrix of $\SG$ is
\begin{equation}
A(\SG) = 
\begin{bmatrix}
-A(G_1) & J_{n_1,n_2} \\
J_{n_2,n_1} & A(G_2) 
\end{bmatrix}
\end{equation}
If $\ket{\alpha} \neq \ket{\jj_{n_1}}$ is an eigenvector of $A(G_1)$ with eigenvalue $\alpha$, then
$\ket{\alpha,0_{n_2}}$ is an eigenvector of $A(\SG)$ with eigenvalue $-\alpha$.
Similarly,
if $\ket{\beta} \neq \ket{\jj_{n_2}}$ is an eigenvector of $A(G_2)$ with eigenvalue $\beta$, then
$\ket{0_{n_1},\beta}$ is an eigenvector of $A(\SG)$ with eigenvalue $\beta$.
The two remaining eigenvalues of $A(\SG)$ are the solutions of the quadratic equation
\begin{equation}
\lambda^{2} + (k_1 - k_2)\lambda - (k_1 k_2 + n_1 n_2) = 0.
\end{equation}
So, $\lambda_{\pm} = \tdm \pm \Delta$ with the corresponding eigenvectors
$\ket{x_{\pm} \jj_{n_1}, y_{\pm} \jj_{n_2}}$, where the two non-zero 
constants $x$ and $y$ are related through the equations
\begin{equation}
(\lambda + k_1) x = n_2 y, 
\ \ \
(\lambda - k_2) y = n_1 x.
\end{equation}
Letting $y = 1$, we get the normalized eigenvectors
\begin{equation}
\ket{\lambda_{\pm}} = \frac{1}{\sqrt{L_{\pm}}} \ket{x_{\pm} \jj_{n_1}, \jj_{n_2}},
\end{equation}
where $x_{\pm} = (\lambda_{\pm} - k_2)/n_1$ and $L_{\pm} = n_1 x_{\pm}^{2} + n_2$.

Suppose $A(G_1) = \sum_{\alpha} \alpha E_{\alpha}$ is the spectral decomposition of $A(G_1)$.
The quantum walk on $\SG$ from vertex $a$ to $b$ (within the $G_1$ subgraph of $\SG$) is given by
\begin{equation}
\bra{b,0}e^{-itA(\SG)}\ket{a,0} =
	\sum_{\alpha \neq k_1} e^{it\alpha}\bra{b}E_{\alpha}\ket{a} + 
	\sum_{\pm} \frac{x_{\pm}^{2}}{L_{\pm}} e^{-it\lambda_{\pm}}.
\end{equation}
By completing the first term and simplifying the second, this yields
\begin{equation}
\bra{b,0}e^{-itA(\SG)}\ket{a,0} =
	\bra{b}e^{itA(G_1)}\ket{a} - \frac{e^{it k_1}}{n_1} + 
	e^{-it\tdm} \left[ \frac{1}{L_{+}L_{-}} \sum_{\pm} e^{\mp it\Delta} x_{\pm}^{2}L_{\mp}  \right]
\end{equation}
The last term may be simplified further using the following identities:
\begin{eqnarray}
\sum_{\pm} x_{\pm} 	& = & \frac{2\tdp}{n_1}, \\
\prod_{\pm} x_{\pm} & = & -\frac{n_2}{n_1}, \\
\prod_{\pm} L_{\pm} & = & 4\Delta^{2} \ \frac{n_2}{n_1}, \\
x_{\pm}^{2}L_{\mp} 	& = & L_{\pm} \ \frac{n_2}{n_1}.
\end{eqnarray}
Thus, we have
\begin{equation}
\bra{b,0}e^{-itA(\SG)}\ket{a,0} =
	\bra{b}e^{itA(G_1)}\ket{a} - \frac{e^{it k_1}}{n_1} + 
	\frac{e^{-it\tdm}}{n_1} \left(\cos(t\Delta) - i\frac{\tdp}{\Delta}\sin(t\Delta)\right).	
\end{equation}
By combining the last two terms using $k_1 + \tdm = -\tdp$, we obtain the claim.
\end{proof}

\par\noindent
The following are immediate corollaries of Theorem \ref{thm:join}. 

\begin{corollary}
Let $G_{1}$ be a $(n_{1},k_{1})$-regular graph with perfect state transfer 
from vertex $a$ to vertex $b$ at time $t = \pi/D$, for some positive integer $D$.
Let $G_{2}$ be a $(n_{2},k_{2})$-regular graph.
Then, the signed graph $G^{-} + H^{+}$ has perfect state transfer from vertex $a$ to vertex $b$
at time $t = \pi/D$ if one of the following conditions hold:
\begin{itemize}
\item $\Delta \equiv 0\pmod{2D}$ and $k_{1}+k_{2} \equiv 0\pmod{4D}$;
\item $\Delta \equiv D\pmod{2D}$ and $k_{1}+k_{2} \equiv 2D\pmod{4D}$,
\end{itemize}
where $\Delta = \frac{1}{2}\sqrt{(k_{1}+k_{2})^{2} + 4n_{1}n_{2}}$.
\end{corollary}

\vspace{.1in}
\par\noindent
Another immediate corollary of Theorem \ref{thm:join} is the following curious result which
shows that the perfect state transfer time within a $2$-clique can be made arbitrarily 
small by joining it with a suitably large $3$-regular graph; see Figure \ref{fig:signed-k6}.

\begin{corollary}
Let $G$ be a $3$-regular graph on $n$ vertices.
Then, the signed graph $K_{2}^{-} + G^{+}$ has perfect state transfer between 
the two vertices of $K_{2}$ at time $\pi/\Delta$ where $\Delta = \sqrt{4+2n}$.
\end{corollary}

\vspace{.1in}
\par\noindent{\em Remark}:
The above result does not hold on the unsigned join $K_{2} + G$ if $G$ is a $(n,3)$-regular graph.
For a $(n,k)$-regular graph $G$, Angeles-Canul \etal \cite{anoprt10} proved that $K_{2} + G$ has
perfect state transfer between the vertices of the $2$-clique provided
$\Delta = \sqrt{(k-1)^{2}+8n}$ is an integer and that both $k-1$ and $\Delta$ are divisible by $8$.
The last condition is clearly impossible when $k=3$.


\section{Decomposition}

In this section, we exploit the fact that the adjacency matrix of a signed graphs may be decomposed 
into positive and negative parts. This decomposition defines the positive and negative subgraphs of 
a signed graph. We describe some results on perfect state transfer on signed graphs
under certain assumptions on these subgraphs.

\begin{theorem} \label{thm:decomposition}
Let $G=(V,E)$ be a graph with perfect state transfer from vertex $a$ to vertex $b$ at time $t$.
Suppose $H$ is a spanning subgraph of $\overline{G}$ and that $H$ is periodic at vertex $a$ with time $t$.
Then, the signed graph $\SG = G^{+} \cup H^{-}$ has perfect state transfer from $a$ to $b$ at time $t$
provided $A(G)$ and $A(H)$ commute.
\end{theorem}
\begin{proof}
Suppose $G$ has perfect state transfer from vertex $a$ to vertex $b$ at time $t$
and $H$ is periodic at vertex $a$ with time $t$, where
$e^{-itA(H)}\ket{a} = e^{i\phi}\ket{a}$ for some real number $\phi$.
The adjacency matrix of $\SG = G^{+} \cup H^{-}$ is given by $A(\SG) = A(G) - A(H)$.
Then, the quantum walk on $\SG$ is
\begin{equation}
\bra{b}e^{-itA(\SG)}\ket{a} = \bra{b}e^{-itA(G)}e^{itA(H)}\ket{a} = e^{i\phi}\bra{b}e^{-itA(G)}\ket{a}.
\end{equation}
This proves the claim.
\end{proof}

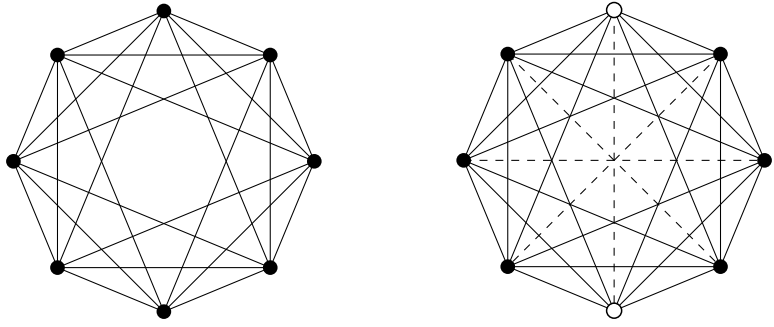
\begin{figure}[t]
\begin{center}
\begin{tikzpicture}[scale=2.0]
\foreach \x in {45,90,...,360}
    \node at (\x:1)[circle,fill=black][scale=0.5] {};

\foreach \x in {45,90,...,360}
    \foreach \y in {45,90,135}
    {
        \draw (\x:1)--(\x+\y:1);
    }
\end{tikzpicture}	\quad	\quad	\quad	\quad
\begin{tikzpicture}[scale=2.0]
\foreach \x in {45,90,...,360}
    \node at (\x:1)[circle,fill=black][scale=0.5] {};

\foreach \x in {45,90,...,360}
    \foreach \y in {45,90,135}
    {
        \draw (\x:1)--(\x+\y:1);
    }

\foreach \x in {45,90,...,180}
{
	\draw[dashed] (\x:1)--(\x+180:1);
}

\foreach \x in {90,270}
{
    \node at (\x:1)[circle,fill=white][scale=0.5] {};
    \draw[line width=0.2mm] (\x:1) circle (0.05cm);
}
\end{tikzpicture}
\caption{More on signed complete graphs. 
(a) $K_{8}$ has no perfect state transfer.
(b) The signed complete graph $K_{8}^{\pm}$ has perfect state transfer between vertices marked white.
	Dashed edges are negatively signed.
}
\label{fig:signed-k8}
\end{center}
\hrule
\end{figure}

\vspace{.1in}
\par\noindent
In the following result, we apply Theorem \ref{thm:decomposition} to signed complete graphs.

\begin{corollary} \label{cor:srg}
Let $G$ be a 
regular graph which has perfect state transfer from vertex $a$ to vertex $b$ at time $t$.
Suppose that $\overline{G}$ is periodic at vertex $a$ at time $t$.
Then, the signed complete graph $\SG = G^{+} \cup \overline{G}^{-}$
has perfect state transfer from vertex $a$ to vertex $b$ at time $t$.
\end{corollary}
\begin{proof}
Since $G$ is regular, $A(G)$ commutes with $J$.
The adjacency matrix of $\overline{G}$ is $A(\overline{G}) = J-I-A(G)$ 
which clearly commutes with $A(G)$. Thus, we may apply Theorem \ref{thm:decomposition}.
\end{proof}

\vspace{.1in}
\par\noindent{\em Remark}:
Consider the $n$-partite graph $G = K_{2,2,\ldots,2}$ 
(also known as the cocktail party graph)
which has antipodal perfect state transfer between vertices in the same partition at time $\pi/2$
(see Ba\v{s}i\'{c} and Petkovic \cite{bp09} and Angeles-Canul \etal \cite{anoprt10}).
If we view this graph as a circulant with $N=2n$ vertices, then let $P_{\sigma}$ denote 
the permutation $\sigma: x \mapsto x+N/2\pmod{N}$. Note that $\sigma$ is an automorphism of $G$.
By Corollary \ref{cor:srg}, we have that the signed complete graph 
\begin{equation}
K_{N}^{\pm} = G^{+} \cup P_{\sigma}^{-}
\end{equation}
has perfect state transfer between vertices $x$ and $x+N/2$ at time $\pi/2$.
Recall that no {\em unsigned} complete graph $K_{N}$ 
has perfect state transfer for $N \ge 3$; see Figure \ref{fig:signed-k8}.
\vspace{.1in}

\par\noindent
Next, we use some useful constructions of perfect state transfer and periodic graphs
from {\em cubelike} graphs which are Cayley graphs over the abelian group $\ZZ_{2}^{d}$.

\begin{theorem} (Cheung and Godsil \cite{cg11}) \\
Let $C \subseteq \ZZ_{2}^{d}$ and let $\delta$ be the sum of the elements of $C$.
If $\delta \neq 0$, then the Cayley graph $G = X(\ZZ_{2}^{d},C)$ has perfect state transfer from $u$ to $u + \delta$
at time $\pi/2$, for each $u \in \ZZ_{2}^{d}$. If $\delta = 0$, then $G$ is periodic with period $\pi/2$.
\end{theorem}

\par\noindent{\em Remark}:
Let $G = X(\ZZ_{2}^{d},C)$ be a Cayley graph where $\delta$, the sum of elements of $C$, is nonzero.
Suppose $P_{\delta}$ is the $2^{d} \times 2^{d}$ permutation matrix representing the
bijection $x \mapsto x+\delta$. Then, by Theorem \ref{thm:decomposition}, the signed cubelike graph 
\begin{equation}
\Sigma = G^{+} \cup P_{\delta}^{-}
\end{equation}
has perfect state transfer from $u$ to $u + \delta$ at time $\pi/2$, for each $u \in \ZZ_{2}^{d}$.
But, the {\em unsigned} cubelike graph 
\begin{equation}
\GG = G^{+} \cup P_{\delta}^{+}
\end{equation}
is merely periodic with period $\pi/2$; see Figure \ref{fig:cubelike-k8}.

\begin{figure}[t]
\begin{center}
\begin{tikzpicture}[scale=1.5]
\foreach \x in {113,158,...,360}
    \node at (\x:1)[circle,fill=black][scale=0.5] {};

\foreach \x in {23,113,...,360}
    {
        \draw (\x:1)--(\x-45:1);
    }

\foreach \x in {68,113,248,293}
	{
		\draw (\x:1)--(\x+90:1);
	}

\foreach \x in {23,68,...,360}
	{
		\draw (\x:1)--(\x+180:1);
	}

\foreach \z in {0,1,2,3}
	{
		\draw[dashed] (68 + 45*\z:1)--(23 - 45*\z:1);
	}

\foreach \x in {23,68}
{
    \node at (\x:1)[circle,fill=white][scale=0.5] {};
    \draw[line width=0.2mm] (\x:1) circle (0.07cm);
}

\node at (68:1.35)  {$000$};
\node at (113:1.35) {$001$};
\node at (158:1.35) {$010$};
\node at (203:1.35) {$011$};
\node at (248:1.35) {$100$};
\node at (293:1.35) {$101$};
\node at (338:1.35) {$110$};
\node at (383:1.35) {$111$};
\end{tikzpicture} 	\quad	\quad	\quad	\quad
\begin{tikzpicture}[scale=1.5]
\foreach \x in {23,68,...,360}
    \node at (\x:1)[circle,fill=black][scale=0.5] {};

\foreach \x in {23,113,...,360}
    {
        \draw (\x:1)--(\x-45:1);
    }

\foreach \x in {68,113,248,293}
	{
		\draw (\x:1)--(\x+90:1);
	}

\foreach \x in {23,68,...,360}
	{
		\draw (\x:1)--(\x+180:1);
	}

\foreach \z in {0,1,2,3}
	{
		\draw (68 + 45*\z:1)--(23 - 45*\z:1);
	}

\node at (68:1.35)  {$000$};
\node at (113:1.35) {$001$};
\node at (158:1.35) {$010$};
\node at (203:1.35) {$011$};
\node at (248:1.35) {$100$};
\node at (293:1.35) {$101$};
\node at (338:1.35) {$110$};
\node at (383:1.35) {$111$};
\end{tikzpicture}
\caption{(a) The cubelike graph $Q_{3}^{+} \cup P_{111}^{-}$ has perfect state transfer
between vertices marked white; dashed edges are negatively signed.
(b) The (unsigned) cubelike graph $Q_{3} \cup P_{111}$ is only periodic.
Here, $P_{111}$ represents the automorphism $x \mapsto x+111$.
}
\label{fig:cubelike-k8}
\end{center}
\hrule
\end{figure}
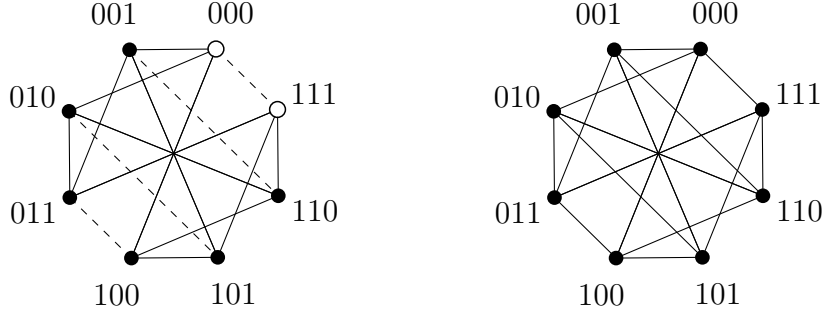


\paragraph{Double Covers}
Given a signed graph $\SG = (G,\sigma)$ where $G=(V,E)$, its {\em double cover} 
is an unsigned graph $\GG$ with vertex set $V \times \{0,1\}$ where 
for each $(u,v) \in E$ and $b \in \{0,1\}$:
\begin{itemize}
\item $(u,b)$ is adjacent to $(v,b)$ if $\sigma(u,v) = +1$; and
\item $(u,b)$ is adjacent to $(v,1-b)$ if $\sigma(u,v) = -1$.
\end{itemize}
Alternatively, if the signed graph $\SG$ has the decomposition $G^{+} \cup G^{-}$, 
its double cover $\GG$ is a graph whose adjacency matrix is
\begin{equation}
A(\GG) = A(G^{+}) \otimes I + A(G^{-}) \otimes X.
\end{equation}

\begin{theorem} \label{thm:double-cover}
Let $\SG = G^{+} \cup G^{-}$ be a signed graph where $A(G^{+})$ and $A(G^{-})$ commute.
Suppose that $G^{+}$ has perfect state transfer from vertex $a$ to vertex $b$ at time $t$ 
and that $G^{-}$ satisfies
\begin{equation}
\bra{b}\cos(A(G^{-})t)\ket{b} = \pm 1.
\end{equation}
Then, the double cover of $\SG$ has perfect state transfer from $(a,1)$ to $(b,1)$ at time $t$.
\end{theorem}
\begin{proof}
Let $\GG$ be the double cover of $\SG$.
Assume that $G^{+}$ has perfect state transfer from vertex $a$ to vertex $b$ at time $t$.
More specifically, suppose that for a real number $\phi$, we have
\begin{equation}
e^{-itA(G^{+})}\ket{a} = e^{i\phi}\ket{b}
\end{equation}
Therefore, the quantum walk on $\GG$ from vertex $a$ to vertex $b$ is given by
\begin{eqnarray}
\bra{b,1}e^{-itA(\GG)}\ket{a,1} 
	& = & \bra{b,1} e^{-it A(G^{-}) \otimes X} e^{-it A(G^{+}) \otimes I} \ket{a,1} \\
	& = & e^{i\phi} \bra{b,1} e^{-it A(G^{-}) \otimes X} \ket{b,1} \\
	& = & e^{i\phi} \bra{b,1} \left[\cos(tA(G^{-})) \otimes I - i \sin(tA(G^{-})) \otimes X \right] \ket{b,1} \\
	& = & e^{i\phi} \bra{b} \cos(tA(G^{-})) \ket{b}.
\end{eqnarray}
This yields the claim.
\end{proof}

\vspace{.1in}
\par\noindent{\em Remark}:
Bernasconi \etal \cite{bgs08} showed that the cubelike graph $G_{1} = X(\ZZ_{2}^{d},C_{1})$ 
has perfect state transfer from $u$ to $u+\delta_{1}$ at time $\pi/2$, where $\delta_{1}$ is 
the sum of all elements of $C_{1}$, provided $\delta_{1} \neq 0$.
When the sum of the generating elements is zero, Cheung and Godsil \cite{cg11} 
proved that there are cubelike graphs $G_{2} = X(\ZZ_{2}^{d},C_{2})$ that have 
perfect state transfer at time $\pi/4$. 
They showed that the latter cubelike graphs correspond to self-orthogonal projective 
binary codes that are even but not doubly even. 
Moreover, these cubelike graphs $G_{2}$ are periodic at time $\pi/2$ and satisfy 
\begin{equation}
\bra{u}\exp(-itA(G_{2}))\ket{u} = e^{-i\pi|C_{2}|/2}.
\end{equation}
See Lemma 3.1 in \cite{cg11} and the comments which followed it.
So, $G_{2}$ is periodic with period $\pm 1$ provided $|C_{2}|$ is even.
Recall that the adjacency matrices of any two cubelike graphs commute since 
they share the same set of eigenvectors (namely, the columns of the Hadamard matrices).
By Theorem \ref{thm:double-cover}, the {\em double cover} of the signed (multi)graph 
\begin{equation}
\Sigma = G_{1}^{+} \cup G_{2}^{-} 
\end{equation}
has perfect state transfer at time $\pi/2$ from $u$ to $u+\delta_{1}$ for each $u \in \ZZ_{2}^{d}$.


\section{Signed Quotients}

Given a graph $G=(V,E)$, a vertex partition $\pi$ given by $V = \biguplus_{k=1}^{m} V_{k}$
is called an {\em equitable partition} of $G$ if for each $j,k$ there are constants $d_{j,k}$ so that
the number of neighbors in $V_{k}$ of each vertex in $V_{j}$ is $d_{j,k}$; this is independent of the 
choice of the vertex of $V_{j}$. That is, for each $x \in V_{j}$ we have
\begin{equation}
d_{j,k} = |N(x) \cap V_{k}|,
\end{equation}
where $N(x) = \{y : (x,y) \in E\}$ is the set of neighbors of $x$.
When the context is clear, we use $\pi_{j}$ in place of $V_{j}$ and
$\pi(u)$ to denote the partition which contains vertex $u$.
The size of the equitable partition is denoted $|\pi|=m$.
Let $P_{\pi}$ be the partition matrix of $\pi$ defined as $\bra{x}P_{\pi}\ket{\pi_{k}} = \iverson{x \in \pi_{k}}$.
It is more useful to work with the normalized partition matrix $Q_{\pi}$ defined as:
\begin{equation}
Q_{\pi} = \sum_{k=1}^{m} \frac{1}{\sqrt{|\pi_{k}|}} P_{\pi}\ketbra{\pi_{k}}{\pi_{k}}.
\end{equation}
For a signed graph $\SG=G^{+} \cup G^{-}$, we say a vertex partition $\pi = \biguplus_{k=1}^{m} \pi_{k}$ 
is equitable for $\SG$ if $\pi$ is an equitable partition for both $G^{+}$ and $G^{-}$.
We use $d_{j,k}^{+}$ and $d_{j,k}^{-}$ to denote $d_{j,k}$ restricted to $G^{+}$ and $G^{-}$, respectively.
Also, we let $d_{j,k}^{\pm} = d_{j,k}^{+} - d_{j,k}^{-}$.
Let $A(\SG/\pi)$ be a symmetric $m \times m$ matrix with rows and columns indexed by the partitions of $\pi$
and whose entries are defined by
\begin{equation}
\bra{\pi_{j}}A(\SG/\pi)\ket{\pi_{k}} = \frac{d^{\pm}_{j,k}}{|d^{\pm}_{j,k}|} \sqrt{|d^{\pm}_{j,k} d^{\pm}_{k,j}|}.
\end{equation}
Here, $A(\SG/\pi)$ is the adjacency matrix of a weighted signed graph $\SG/\pi$ 
(which we call the quotient of $\SG$ modulo $\pi$). 

\vspace{.1in}
\par\noindent
The next lemma generalizes a result for unsigned graphs (see Godsil \cite{godsil-dm11}).

\begin{lemma} \label{lemma:equitable}
Let $\Sigma = G^{+} \cup G^{-}$ be a signed graph and $\pi$ be an equitable partition of $\Sigma$ 
with a normalized partition matrix $Q_{\pi}$. Then:
\begin{enumerate}
\item $Q_{\pi}^{T}Q_{\pi} = I_{|\pi|}$.
\item $Q_{\pi}Q_{\pi}^{T} = \diag(|\pi_{k}|^{-1} J_{|\pi_{k}|})$.
\item $Q_{\pi}Q_{\pi}^{T}$ commutes with $A(\SG)$.
\item $A(\SG/\pi) = Q_{\pi}^{T} A(\SG) Q_{\pi}$.
\end{enumerate}
\end{lemma}
\begin{proof}
The first two properties hold since the columns of $Q_{\pi}$ are the normalized characteristic vectors 
of the partition matrix $P_{\pi}$.
The third property holds since for any $a$ and $b$ we have
\begin{equation}
\bra{b}A(\SG)Q_{\pi}Q_{\pi}^{T}\ket{a}
	= \frac{1}{|\pi(a)|} \sum_{u \in \pi(a)} \bra{b}A(\Sigma)\ket{u} 
	= \frac{1}{|\pi(b)|} \sum_{v \in \pi(b)} \bra{v}A(\Sigma)\ket{a}
	= \bra{b}Q_{\pi}Q_{\pi}^{T}A(\Sigma)\ket{a}.
\end{equation}
Given that $\pi$ is equitable for both $G^{+}$ and $G^{-}$, we have
\begin{equation}
d^{\pm}_{j,k}|\pi_{j}| = d^{\pm}_{k,j}|\pi_{k}|.
\end{equation}
Since $A(\SG) = A(G^{+}) - A(G^{-})$, we have
\begin{equation}
\bra{\pi_{j}}Q_{\pi}^{T} A(\SG) Q_{\pi}\ket{\pi_{k}}
	= \frac{\bra{\pi_{j}}P_{\pi}^{T}A(\SG)P_{\pi}\ket{\pi_{k}}}{\sqrt{|\pi_{j}||\pi_{k}|}} 
	= \frac{d^{\pm}_{j,k}}{\sqrt{|\pi_{j}||\pi_{k}|}} 
	= \pm \sqrt{|d^{\pm}_{j,k}d^{\pm}_{k,j}|},
\end{equation}
where the sign is $+1$ if $d^{\pm}_{j,k} > 0$, and $-1$ otherwise.
This proves property the last property.
\end{proof}

\vspace{.1in}
\par\noindent
The following theorem generalizes a result in Bachman \etal \cite{bfflott12} on
the equivalence of perfect state transfer on unsigned graphs and their quotients.
For completeness, we provide the proof which follows from Lemma \ref{lemma:equitable}.

\begin{theorem} \label{thm:pst-quotient}
Let $\SG$ be a signed graph with an equitable partition $\pi$ 
where vertices $a$ and $b$ belong to singleton cells.
Then, for any time $t$
\begin{equation}
\bra{b}\exp(-itA(\SG))\ket{a} = \bra{\pi(b)}\exp(-itA(\SG/\pi))\ket{\pi(a)}.
\end{equation}
\end{theorem}
\begin{proof} 
Since $A(\SG)$ commutes with $QQ^{T}$, we have $(QQ^{T}A(\SG))^{k} = A(\SG)^{k}QQ^{T}$ for $k \ge 1$.
Given that $a$ and $b$ are in singleton cells, $\ket{\pi(a)} = Q^{T}\ket{a}$ and $\ket{\pi(b)} = Q^{T}\ket{b}$. 
Thus, we have
\begin{eqnarray}
\bra{\pi(b)}e^{-itA(\SG/\pi)}\ket{\pi(a)}
	& = & \bra{\pi(b)} e^{-it Q^{T}A(\SG)Q} \ket{\pi(a)} \\
	& = & \bra{b}Q \left[ \sum_{k=0}^{\infty} \frac{(-it)^{k}}{k!} (Q^{T}A(\SG)Q)^{k} \right] Q^{T}\ket{a} \\
	& = & \bra{b} \left[ \sum_{k=0}^{\infty} \frac{(-it)^{k}}{k!} (QQ^{T}A(\SG))^{k} \right] QQ^{T}\ket{a} \\
	& = & \bra{b} e^{-it A(\SG)} QQ^{T}\ket{a},
\end{eqnarray}
which proves the claim since $QQ^{T}\ket{a} = \ket{a}$ because $a$ belongs to a singleton cell.
\end{proof}


\section{Many-Particle Quantum Walks}

In this section, we describe graph operators which arise naturally from $k$-particle quantum walk 
on an underlying graph $G$. These operators include symmetric powers, Cartesian quotients, 
and {\em signed} exterior powers.
\begin{itemize}
\item
The weighted {\em Cartesian quotients} were described by Feder \cite{f06} as a generalization 
of graphs studied by Christandl \etal \cite{cdel04,cddekl05}. The construction in \cite{f06}
is explicitly based on many-boson quantum walks on graphs.
Subsequently, Bachman \etal \cite{bfflott12} showed that these graphs $\symp{G}{k}$ described
by Feder are quotients of $k$-fold Cartesian product modulo a natural equitable partition.

\item
The {\em symmetric powers} $\veep{G}{k}$ were studied by Audenaert \etal \cite{agrr07}
in the context of graph isomorphism on strongly regular graphs.
These graphs are equivalent to the $k$-tuple vertex graphs introduced by Zhu \etal \cite{zlla92}
and were were later studied by Osborne \cite{o06} in connection with quantum spin networks.
It can be shown that these graphs are based on many-particle quantum walks with hardcore bosons.

\item
The {\em exterior powers} $\wedgep{G}{k}$ and their perfect state transfer properties as {\em signed} graphs
will be our main focus in this section. We formally define these graphs in the next section and connect them
with many-fermion quantum walks.
\end{itemize}
A common trait shared by these families of graphs is that they are derived from a $k$-fold Cartesian product
$G^{\Box k}$ of an underlying graph $G$. The $k$-fold Cartesian product represents, in a natural way,
a $k$-particle quantum walk on the graph $G$ where the particles are distinguishable.
Here, the edges of the Cartesian product represent scenarios where a single particle hops
from its current vertex to a neighboring one.
Since quantum particles are either fermions or bosons, this influences the resulting graph-theoretic 
construction in an interesting manner.

In the fermionic regime, no two particle may occupy the same vertex due to the Pauli exclusion principle.
Mathematically, this requires the removal of a ``diagonal'' set $\mathcal{D}$
consisting of all $k$-tuples with a repeated vertex.
On the other hand, in the bosonic regime, the particles are allowed to occupy the same vertex.
But in the so-called {\em hardcore bosonic} model, we are also required to remove the diagonal
since these bosons are not allowed to occupy the same vertex
although they lack the antisymmetric exchange property present in fermions.

Since quantum particles are indistinguishable, the final step in these constructions ``collapses''
certain configurations of the particles in the $k$-fold Cartesian product. Mathematically, this
simply involves a conjugation by either the symmetrizer (bosonic) or the anti-symmetrizer (fermionic);
more specifically, this is a projection to either the symmetric or anti-symmetric subspace of
the underlying $k$-fold tensor product space (see Bhatia \cite{bhatia}).

Therefore, the full construction of a graph $\GG$ based on a $k$-particle quantum walk on $G$
may be described via its adjacency matrix as:
\begin{equation} \label{eqn:quotient}
A(\GG) = \mathcal{Q}^{\dagger} A(G^{\Box k} \setminus \mathcal{D}) \mathcal{Q}.
\end{equation}
In Equation (\ref{eqn:quotient}), the matrix $\mathcal{Q}$ is
the symmetrizer $P_{\vee}$ in the case of Cartesian quotients,
the anti-symmetrizer $P_{\wedge}$ in the case of exterior powers,
or a hybrid in the case of symmetric powers.
Here, we borrow the notation $P_{\vee}$ and $P_{\wedge}$ from Bhatia \cite{bhatia}
(see also Audenaert \etal \cite{agrr07}).
The diagonal set $\mathcal{D}$ is removed in all constructions except for 
Cartesian quotients; for exterior powers, this diagonal removal is implicitly done 
via the anti-symmetrizer $P_{\wedge}$.
Overall, we will view $\GG$ as a ``quotient'' of $G^{\Box k}$ modulo a suitable 
equitable partition. 

\paragraph{Cartesian quotients}
For a graph $G$ on $n$ vertices and a positive integer $k$,
the Cartesian quotient $\symp{G}{k}$ is a graph whose vertex set is the
set of $n$-tuple of non-negative integers whose sum is $k$, that is,
\begin{equation}
V(\symp{G}{k}) = \left\{a \in \NN^{V(G)}: \sum_{v \in V(G)} a_{v} = k\right\},
\end{equation}
and whose edges are the pairs $(a,b)$ which satisfy
\begin{equation}
(\exists (u,v) \in E(G))[b_{u}=a_{u}-1 \ \wedge \ b_{v}=a_{v}+1 \ \wedge \ (\forall w \not\in \{u,v\})[a_{w}=b_{w}]]
\end{equation}
with an edge weight of $\sqrt{(a_{u}-1)(a_{v}+1)}$.
This construction corresponds to a $k$-{\bf boson} quantum walk on $G$ 
where the adjacency matrix of $\symp{G}{k}$ is given by
\begin{equation}
A(\symp{G}{k}) = P_{\vee}^{\dagger} A(G^{\Box k}) P_{\vee},
\end{equation}
and the resulting graph is weighted and unsigned.
This class of graphs was studied by Feder \cite{f06} and then by Bachman \etal \cite{bfflott12}.

\begin{figure}[t]
\begin{center}
\begin{tabular}{|c|c|c|c|} \hline
{Graph}                            & {Notation}  &   {Particle}       & {Citation} \\ \hline \hline
{\em weighted} Cartesian quotients & $\symp{G}{k}$  &   bosonic             & Feder \cite{f06}, Bachman \etal \cite{bfflott12} \\ \hline
symmetric powers & $G^{\{k\}}$                      &   hardcore bosonic    & Audenaert \etal \cite{agrr07} \\
$k$-tuple vertex graphs & $U_{k}(G)$                    &                       & Zhu \etal \cite{zlla92}, Osborne \cite{o06} \\ \hline
{\em signed} exterior powers & $\wedgep{G}{k}$      &   fermionic           & this work \\ \hline
\end{tabular}
\caption{Graph operators from many-particle quantum walks on graphs.}
\end{center}
\end{figure}

\paragraph{Symmetric powers}
For a graph $G$ on $n$ vertices and a positive integer $k$ where $1 \le k \le n$,
the symmetric $k$th power $\veep{G}{k}$ of $G$ 
is a graph whose vertices are the $k$-subsets of $V(G)$ and whose edges consist of pairs of
$k$-subsets $(A,B)$ for which $A \symdif B \in E(G)$.
This construction corresponds to a $k$-{\bf hardcore boson} quantum walk on $G$ where
bosons are not allowed to occupy the same vertex but without the fermionic antisymmetric exchange rule.
Here, we have
\begin{equation}
A(\veep{G}{k}) = (P^{(k)})^{\dagger} A(G^{\Box k} \setminus \mathcal{D}) P^{(k)},
\end{equation}
where $P^{(k)}$ is a ``signless'' variant of $P_{\wedge}$ (see Audenaert \etal \cite{agrr07})
and $\mathcal{D}$ consists of $k$-tuples from $V^{k}$ which contain a repeated vertex.
This class of unweighted and unsigned graphs was studied by Audenaert \etal \cite{agrr07}
and also by Zhu \etal \cite{zlla92} and Osborne \cite{o06}.

\paragraph{Exterior powers}
For a graph $G$ on $n$ vertices and a positive integer $k$ where $1 \le k \le n$,
the exterior $k$th power $\wedgep{G}{k}$ 
is a graph whose vertices are the {\em ordered} $k$-subsets of $V(G)$ and whose edges consist of
pairs of ordered $k$-subsets $(A,B)$ for which $A \symdif B \in E(G)$. Moreover, if $\pi \in S_{k}$
is the unique permutation which yields $A_{\pi(i)} = B_{i}$, for all $i$ except for one $j$
where $(A_{\pi(j)},B_{j}) \in E(G)$, then the edge $(A,B)$ is assigned the sign $\sgn(\pi)$.
This construction is based on a $k$-{\bf fermion} quantum walk on $G$.
Here, we have
\begin{equation}
A(\xwedgep{G}{k}) = P_{\wedge}^{\dagger} A(G^{\Box k}) P_{\wedge},
\end{equation}
where $P_{\wedge}$ implicitly removes a diagonal set consisting of $k$-tuples from $V^{k}$ 
with repeated vertices.
The resulting graph is {\em signed}.
We remark that the graphs $U_{k}(G)$ studied in \cite{zlla92,o06} are simply
the exterior powers $\wedgep{G}{k}$ with the signs ignored.

We will define the exterior power $\wedgep{G}{k}$ more formally in the following section.

\subsection{Exterior Powers}

Given a graph $G=(V,E)$ and a positive integer $k$, where $1 \le k \le |V|-1$, 
the {\em exterior $k$th power} of $G$, which is denoted $\wedgep{G}{k}$, is defined as follows. 
Suppose $V = \{v_{1},\ldots,v_{n}\}$ is the vertex set of $G$.
We fix some (arbitrary) total ordering $\prec$ on $V$, say, $v_{1} \prec \ldots \prec v_{n}$. 
Let $\binom{V}{k}$ be the $k$-subsets of $V$ where each subset is ordered according to $\prec$.
That is, we have
\begin{equation}
\binom{V}{k} = \{(u_{1},\ldots,u_{k}) \in V^{k} : u_{1} \prec \ldots \prec u_{k}\}.
\end{equation}
The elements of $\binom{V}{k}$ are often denoted as a ``wedge'' product
$u_{1} \wedge \ldots \wedge u_{k}$, again under the condition that $u_{1} \prec \ldots \prec u_{k}$.

The vertex set of $\wedgep{G}{k}$ is $\binom{V}{k}$.
The edge set of $\wedgep{G}{k}$ consists of pairs $\bigwedge_{j=1}^{k} u_{j}$ and $\bigwedge_{j=1}^{k} v_{j}$ 
for which there is a bijection $\pi$ over $[n]$ so that $u_{\pi(j)} = v_{j}$ for all $j$ 
except at one index $i$ where $(u_{\pi(i)},v_{i}) \in E(G)$;
most importantly, the {\em sign} of this edge is defined to be $\sgn(\pi)$.
Thus, $\wedgep{G}{k}$ is a {\em signed} graph on $\binom{n}{k}$ vertices.
In summary, we have
\begin{eqnarray}
V(\xwedgep{G}{k}) & = & \binom{V}{k} \\
E(\xwedgep{G}{k}) & = &
	\left\{(\mbox{$\bigwedge_{j=1}^{k} u_{j}$}, \mbox{$\bigwedge_{j=1}^{k} v_{j}$}) : 
		(\exists i)[(u_{\pi(i)},v_{i}) \in E \ \wedge \ (\forall j \neq i)[u_{\pi(j)} = v_{j}]]\right\}
\end{eqnarray}

\vspace{.1in}
\par\noindent{\em Example}: 
Let $V = \{a,b,c,d\}$ be ordered as $a \prec b \prec c \prec d$. Then, 
\begin{equation}
\binom{V}{2} = \{a \wedge b, a \wedge c, a \wedge d, b \wedge c, b \wedge d, c \wedge d\}.
\end{equation}
Suppose $G$ is a $4$-cycle defined on $V$ where $a$ and $d$ are not adjacent but are both adjacent 
to $b$ and $c$; see Figure \ref{fig:ext-c4}.
Here, $\wedgep{G}{2}$ is a signed $K_{2,4}$ with the bipartition given by
$\{a \wedge d, b \wedge c\}$ and $\{a \wedge b, a \wedge c, b \wedge d, c \wedge d\}$.
All edges are positive except for the edges between $a \wedge b$ and $b \wedge c$ 
(since this requires a transposition to ``align'' the common vertex $b$)
and between $c \wedge d$ and $b \wedge c$ (by a similar reasoning). 
\vspace{.2in}

\begin{figure}[t]
\begin{center}
\begin{tikzpicture}
%
\draw[line width=0.25mm] (0,-1)--(-1,0)--(0,1)--(1,0)--(0,-1);

\foreach \y in {-1,1}
    \node at (0,\y)[circle, fill=black][scale=0.5]{};
\foreach \x in {-1,1}
    \node at (\x,0)[circle, fill=black][scale=0.5]{};

\node at (0,-1.3)[scale=0.9]{$a$};
\node at (-1.3,0)[scale=0.9]{$b$};
\node at (0,+1.3)[scale=0.9]{$d$};
\node at (+1.3,0)[scale=0.9]{$c$};

%
\foreach \y in {-1}
    \foreach \x in {4.25,6.25,7.75,9.75}
        \draw[line width=0.2mm] (7,\y)--(\x,0);

\foreach \y in {+1}
    \foreach \x in {6.25,7.75}
        \draw[line width=0.2mm] (7,\y)--(\x,0);

\foreach \y in {+1}
    \foreach \x in {4.25,9.75}
        \draw[dashed] (7,\y)--(\x,0);

\foreach \x in {4.25,6.25,7.75,9.75}
    \node at (\x,0)[circle, fill=black][scale=0.5]{};
\node at (3.6,0)[scale=0.9]{$a \wedge b$};
\node at (5.6,0)[scale=0.9]{$a \wedge c$};
\node at (8.4,0)[scale=0.9]{$b \wedge d$};
\node at (10.4,0)[scale=0.9]{$c \wedge d$};

\foreach \y in {-1,1}
    \node at (7,\y)[circle, fill=black][scale=0.5]{};
\node at (7,+1.3)[scale=0.9]{$b \wedge c$};
\node at (7,-1.3)[scale=0.9]{$a \wedge d$};



\end{tikzpicture}
\caption{(a) The four-cycle $C_{4}$; (b) Its exterior square $\wedgep{C_{4}}{2}$
under the alphabetic vertex ordering $a \prec b \prec c \prec d$. Dashed edges are negatively signed.
}
\label{fig:ext-c4}
\end{center}
\hrule
\end{figure}
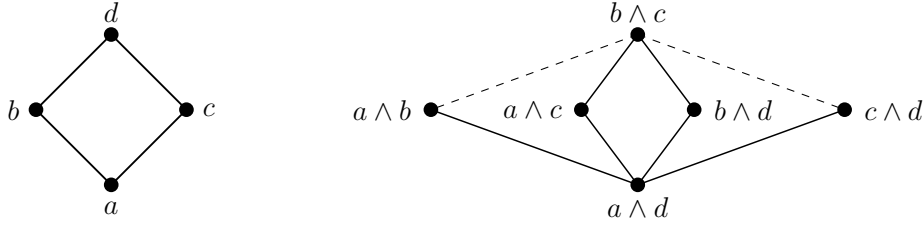

The algebraic view of $\wedgep{G}{k}$ will be more convenient for our purposes.
For a permutation $\pi \in S_{k}$ and a $k$-tuple $u = (u_{1},\ldots,u_{k})$,
we denote $\pi(u)$ to mean the element $(u_{\pi(1)},\ldots,u_{\pi(k)})$.
We consider the {\em anti-symmetrizer} operator $\Alt_{n,k}$ defined as
\begin{equation}
\Alt_{n,k} = \sum_{v \in \binom{V}{k}, \pi \in S_{k}} \frac{\sgn(\pi)}{\sqrt{k!}} \ketbra{\pi(v)}{v}.
\end{equation}
So, $\Alt_{n,k}$ is an injective map from the exterior vector space $\CC^{\binom{V}{k}}$ 
to the tensor product space $\CC^{V^{\otimes k}}$. For each $v \in \binom{V}{k}$, we have
\begin{equation}
\Alt_{n,k}\ket{v} = \sum_{\pi \in S_{k}} \frac{\sgn(\pi)}{\sqrt{k!}} \ket{\pi(v)}.
\end{equation}
Moreover, $\Alt_{n,k}^{T}$ is a surjective map from $\CC^{V^{\otimes k}}$ to $\CC^{\binom{V}{k}}$.
For each $u \in V^{\otimes k}$, 
\begin{equation}
\Alt_{n,k}^{T}\ket{u} = \frac{\sgn(\pi)}{\sqrt{k!}} \ket{v},
\end{equation}
where $v \in \binom{V}{k}$ satisfies $\pi(u) = v$;
here, $\pi$ is the permutation that orders $u$ according to the total order $\prec$.
Next, we show that $\Alt_{n,k}$ defines a ``signed'' equitable partition on $G^{\Box k}$.
This shows that $\wedgep{G}{k}$ is the quotient graph of $G^{\Box k}$ induced by $\Alt_{n,k}$.

\begin{lemma} 
Let $G$ be a graph on $n$ vertices. For each positive integer $k$ with $1 \le k \le n-1$:
\begin{enumerate}
\item $\Alt_{n,k}^{T}\Alt_{n,k} = I_{\binom{n}{k}}$.
\item $\Alt_{n,k} \Alt_{n,k}^{T}$ commutes with $A(G^{\Box k})$.
\item $A(\wedgep{G}{k}) = \Alt_{n,k}^{T} A(G^{\Box k}) \Alt_{n,k}$.
\end{enumerate}
\end{lemma}
\begin{proof}
The first property holds since the columns of $\Alt_{n,k}$ are normalized and form a vertex partition.
For the second property, we note that for each $a \in V^{\otimes k}$:
\begin{equation}
\Alt_{n,k} \Alt_{n,k}^{T} \ket{a} = 
	\frac{1}{k!} \sum_{\pi \in S_{k}} \sgn(\pi) P_{\pi}\ket{a},
\end{equation}
where $P_{\pi}$ denotes the permutation matrix which encodes the action of $\pi$ on $V^{\otimes k}$.
Moreover, $P_{\pi} A(G^{\Box k}) = A(G^{\Box k}) P_{\pi}$ for each permutation $\pi \in S_{k}$.
Therefore, for each $a,b \in V^{\otimes k}$ we have
\begin{eqnarray}
\bra{a} \Alt_{n,k} \Alt_{n,k}^{T} A(G^{\Box k}) \ket{b}
	& = & \frac{1}{k!} \sum_{\pi \in S_{k}} \sgn(\pi) \bra{a} P_{\pi} A(G^{\Box k}) \ket{b} \\
	& = & \frac{1}{k!} \sum_{\pi \in S_{k}} \sgn(\pi) \bra{a} A(G^{\Box k}) P_{\pi} \ket{b} \\
	& = & \bra{a} A(G^{\Box k}) \Alt_{n,k} \Alt_{n,k}^{T} \ket{b}.
\end{eqnarray}
The third property holds since for each $a,b \in \binom{V}{k}$ we have
\begin{eqnarray}
\bra{b} \Alt_{n,k}^{T} A(G^{\Box k}) \Alt_{n,k} \ket{a}
    & = & \frac{1}{k!} \sum_{\pi_{1},\pi_{2}} \sgn(\pi_{1} \circ \pi_{2})
            \bra{b} P_{\pi_{1}}^{T} A(G^{\Box k}) P_{\pi_{2}} \ket{a} \\
    & = & \frac{1}{k!} \sum_{\pi_{1},\pi_{2}} \sgn(\pi_{1} \circ \pi_{2})
            \bra{b} A(G^{\Box k}) P_{\pi_{1} \circ \pi_{2}} \ket{a} \\
    & = & \sum_{\pi} \sgn(\pi) \bra{b} A(G^{\Box k}) P_{\pi} \ket{a}.
\end{eqnarray}
The last expression is either $0$ or $\pm 1$ and agrees with
the definition of adjacency on $\wedgep{G}{k}$.
\end{proof}

\vspace{.1in}
\par\noindent
Next, we show that if a graph $G$ has $k$ disjoint pairs of vertices with perfect state transfer,
then so does the exterior $k$th power of $G$. 

\begin{theorem} \label{thm:pst-ext}
Let $G$ be a graph with perfect state transfer at time $t$ from vertex $a_{j}$ to vertex $b_{j}$, 
for each $j$ with $1 \le j \le k$. 
Suppose that the collection $\{(a_{j},b_{j}) : 1 \le j \le k\}$ are pairwise disjoint.
Then, $\wedgep{G}{k}$ has perfect state transfer from vertex
$\bigwedge_{j=1}^{k} a_{j}$ to vertex $\bigwedge_{j=1}^{k} b_{j}$.
\end{theorem}
\begin{proof}
Let $\hat{a} = \bigwedge_{j=1}^{k} a_{j}$ and $\hat{b} = \bigwedge_{j=1}^{k} b_{j}$.
The quantum walk on $\wedgep{G}{k}$ from $\hat{a}$ to $\hat{b}$ is given by
\begin{eqnarray}
\bra{\hat{b}} \exp\left[-itA(\xwedgep{G}{k})\right] \ket{\hat{a}}
	& = & \bra{\hat{b}} \Alt_{n,k}^{T} e^{-itA(G^{\Box k})} \Alt_{n,k} \ket{\hat{a}} \\
	& = & \frac{1}{k!} 
		\left[ \sum_{\pi_{1}} \sgn(\pi_{1}) \bra{\pi_{1}(\hat{b})} \right]
		e^{-itA(G^{\Box k})} 
		\left[ \sum_{\pi_{2}} \sgn(\pi_{2}) \ket{\pi_{2}(\hat{a})} \right] \\
	& = & \frac{1}{k!} 
		\sum_{\pi_{1},\pi_{2}} \sgn(\pi_{1} \circ \pi_{2}) 
		\bra{\pi_{1}(\hat{b})} e^{-itA(G^{\Box k})} \ket{\pi_{2}(\hat{a})} \\
	& = & \frac{1}{k!} 
		\sum_{\pi} \bra{\pi(\hat{b})} e^{-itA(G^{\Box k})} \ket{\pi(\hat{a})}.
\end{eqnarray}
This proves the claim.
\end{proof}

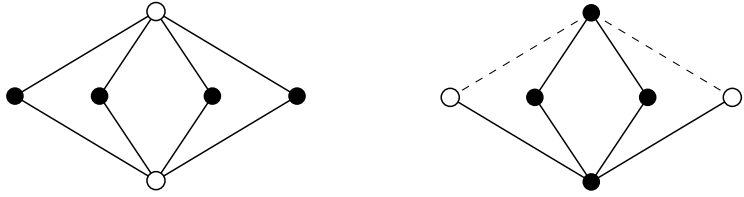
\begin{figure}[t]
\begin{center}
\begin{tikzpicture}[scale=1.5]
\foreach \y in {-0.75,+0.75}
    \foreach \x in {5.75,6.5,7.5,8.25}
        \draw[line width=0.2mm] (7,\y)--(\x,0);

\foreach \y in {-0.75,+0.75}
{
    \node at (7,\y)[circle, fill=white][scale=0.6]{};
    \draw[line width=0.2mm] (7,\y) circle (0.08cm);
}

\foreach \x in {5.75,6.5,7.5,8.25}
    \node at (\x,0)[circle, fill=black][scale=0.6]{};
\end{tikzpicture}	\quad	\quad	\quad	\quad
\begin{tikzpicture}[scale=1.5]
\foreach \y in {-0.75}
    \foreach \x in {5.75,6.5,7.5,8.25}
        \draw[line width=0.2mm] (7,\y)--(\x,0);

\foreach \y in {+0.75}
    \foreach \x in {6.5,7.5}
        \draw[line width=0.2mm] (7,\y)--(\x,0);

\foreach \y in {+0.75}
    \foreach \x in {5.75,8.25}
        \draw[dashed] (7,\y)--(\x,0);

\foreach \x in {6.5,7.5}
    \node at (\x,0)[circle, fill=black][scale=0.6]{};

\foreach \x in {5.75,8.25}
{
    \node at (\x,0)[circle, fill=white][scale=0.6]{};
    \draw[line width=0.2mm] (\x,0) circle (0.08cm);
}

\foreach \y in {-0.75,0.75}
    \node at (7,\y)[circle, fill=black][scale=0.6]{};
\end{tikzpicture}
\caption{Complete bipartite graphs.
(a) The unsigned $K_{2,4}$ has perfect state transfer (since its quotient is $P_{3}$).
(b) The exterior square $\wedgep{C_{4}}{2}$ has perfect state transfer (by Theorem \ref{thm:pst-ext})
but on a different pair of vertices; dashed edges are negatively signed. 
Perfect state transfer occurs between vertices marked white within each case.
}
\label{fig:pst-ext-c4}
\end{center}
\hrule
\end{figure}


\section{Conclusions}

In this work, we studied quantum walks on signed graphs.
Our goal was to understand the effect of negatively signed edges on perfect state transfer.
We explore this question in the cases of Cartesian products and graph joins. 
For the latter, we show that the join of a negative $2$-clique with any $3$-regular graph
has perfect state transfer (unlike the unsigned join); moreover, the perfect state transfer
time improves as the size of the $3$-regular graph increases.
Using the natural spanning-subgraph decomposition of signed graphs, we consider signed graphs
obtained from regular graphs and their complements. This yields signed complete graphs which 
have perfect state transfer; in contrast, unsigned complete graphs do not exhibit 
perfect state transfer but are merely periodic.
Most of our results show that signed graphs from multiple switching-equivalent classes of 
the same underlying graph may simultaneously have perfect state transfer.
Finally, we consider a graph operator called exterior powers which had been studied elsewhere 
(but not in the context of signed graphs). We prove conditions for which the exterior power of 
a graph has perfect state transfer. These exterior powers of graph may be of independent interest 
especially in connection with many particle quantum walk on graphs.


\section*{Acknowledgments}

The research was supported in part by the National Science Foundation grant DMS-1004531
and also by the National Security Agency grant H98230-11-1-0206.
We thank David Feder for helpful comments and for pointing out connection with the hardcore boson model.



\end{document}